%% file: main.tex
\documentclass[a4paper,UKenglish,cleveref, autoref, thm-restate]{lipics-v2021}
\pdfoutput=1 %uncomment to ensure pdflatex processing (mandatatory e.g. to submit to arXiv)
\hideLIPIcs  %uncomment to remove references to LIPIcs series (logo, DOI, ...), e.g. when preparing a pre-final version to be uploaded to arXiv or another public repository

\title{Greed is slow on sparse graphs of oriented valued constraints} %TODO Please add

\author{Artem Kaznatcheev}{Department of Mathematics, and Department of Information and Computing Sciences, Utrecht University, The  Netherlands}{a.kaznatcheev@uu.nl}{https://orcid.org/0000-0001-8063-2187}{}%TODO mandatory, please use full name; only 1 author per \author macro; first two parameters are mandatory, other parameters can be empty. Please provide at least the name of the affiliation and the country. The full address is optional. Use additional curly braces to indicate the correct name splitting when the last name consists of multiple name parts.

\author{Sofia {Vazquez Alferez}}{Department of Mathematics, and Department of Information and Computing Sciences, Utrecht University, The Netherlands}{s.vazquezalferez@uu.nl}{https://orcid.org/0000-0002-1541-8683}{}

\authorrunning{A. Kaznatcheev and S. {Vazquez Alferez}} %TODO mandatory. First: Use abbreviated first/middle names. Second (only in severe cases): Use first author plus 'et al.'

\Copyright{Artem Kaznatcheev and Sofia {Vazquez Alferez}} %TODO mandatory, please use full first names. LIPIcs license is "CC-BY";  http://creativecommons.org/licenses/by/3.0/

 %TODO mandatory: Please choose ACM 2012 classifications from https://dl.acm.org/ccs/ccs_flat.cfm 

 %TODO mandatory; please add comma-separated list of keywords

\relatedversion{} %optional, e.g. full version hosted on arXiv, HAL, or other respository/website
\acknowledgements{We want to thank Martin C. Cooper for his comments which greatly improved the quality of our writing and Melle van Marle for his insightful discussions.}%optional

\nolinenumbers %uncomment to disable line numbering

\makeatletter
\let\c@author\relax
\makeatother

\usepackage{amsmath}
\usepackage{amsthm}
\usepackage{amsfonts}
\usepackage{graphicx}
\usepackage{pgf}
\usepackage{caption}
\usepackage{subcaption}

\usepackage{multirow}%multitable (many rows per column, or columns per row)
\usepackage{colortbl}
\definecolor{negCol}{rgb}{0.6, 0.6, 0.6} %{gray}{0.7}
\definecolor{posCol}{rgb}{0.9, 0.9, 0.9}%{gray}{0.9} %{1, 0.76, 0.6}

\usepackage{tikz}
\usetikzlibrary{arrows,automata}
\usetikzlibrary{positioning}
\usetikzlibrary{calc}
\usetikzlibrary{intersections}
\usetikzlibrary{shapes,fit} %use shapes library if you need ellipse
\usetikzlibrary{math}% Required for inserting images
\usetikzlibrary {graphs, quotes}

\newcounter{darrow}
\newcommand\xrightarrowdashed[1]{%
\stepcounter{darrow}%
\;\begin{tikzpicture}
\node (\thedarrow) {\strut$#1$};
\draw[->,dashed] (\thedarrow.south west) -- (\thedarrow.south east);
\end{tikzpicture}\;%
}

\usepackage{circledsteps}

\let\textcite\cite 

\usepackage[small,full]{complexity}
\newclass{\UEOPL}{UEOPL}
\newlang{\wSAT}{Weighted\;2\text{-}SAT}
\newlang{\VCSP}{VCSP}

\begin{document}

\maketitle

\begin{abstract}
Greedy local search is especially popular for solving valued constraint satisfaction problems (\VCSP{s}).
Since any method will be slow for some \VCSP{s}, we ask:
what is the simplest \VCSP{} on which greedy local search is slow?
We construct a \VCSP{} on $6n$ Boolean variables for which greedy local search takes $7(2^n - 1)$ steps to find the unique peak.
Our \VCSP{} is simple in two ways.
First, it is very sparse: its constraint graph has pathwidth $2$ and maximum degree $3$.
This is the simplest \VCSP{} on which some local search could be slow.
Second, it is `oriented' – there is an ordering on the variables such that later variables are conditionally-independent of earlier ones.
Being oriented allows many non-greedy local search methods to find the unique peak in a quadratic number of steps. 
Thus, we conclude that -- among local search methods -- greed is particularly slow. 
\end{abstract}

\section{Introduction}
We cannot hope for a polynomial time algorithm to solve arbitrary combinatorial optimization problems.
But we still need to try to do something to solve these hard problems.
Given that many hard problems can be defined as finding an assignment $x \in \{0,1\}^n$ that maximizes some pseudo-Boolean function -- that we call a \emph{fitness function} -- many turn to local search as a heuristic method~\cite{LocalSearch_Book1}.
Local search methods start at an initial assignment, then apply a fixed update rule to select a better adjacent assignment to move to until no further improvement is possible. 
Such a sequence of adjacent improving assignments is called an ascent.
One of the most popular update rules is to always select the adjacent assignment that increases fitness by the largest amount, that is, to follow the steepest ascent-- 
this update rule defines greedy local search~\cite{approxBook}.
Since greed is just one of many possible update rules, a natural question arises:
is greed good?
Or more specifically: is greedy local search fast or slow compared to other local search methods?

Since no local search method will be able to solve \emph{all} instances of hard combinatorial optimization problems in a polynomial number of steps, we need a more nuanced criterium for what makes a local search method fast or slow.
We find this nuance in looking at the performance of our methods on subsets of instances of differing simplicity instead of on all instances.
Now, if a method cannot solve particularly simple sets of instances in a polynomial number of steps -- at least when compared to other similar methods -- then we call it slow.

To define our set of all instances and refine that to simple instances, we turn to valued constraint satisfaction problems (\VCSP{s}).
In the language of constraint programming, finding the maximum of a pseudo-Boolean function is the same as solving a Boolean \VCSP. 
Given that even binary Boolean \VCSP{s} can express both \NP-hard and \PLS-complete problems~\cite{PLS_Survey, PLS},
we define -- in \cref{sec:background} -- our set of all instances as the set of all binary Boolean \VCSP{s}.
Each binary \VCSP{} can be interpreted as defining a constraint graph with an edge between any two variables that share a constraint.
This allows us to use the sparsity of the resulting graphs as a measure of simplicity: 
we focus on sets of instances of bounded degree and of bounded treewidth or -- more restrictively -- pathwidth.

In terms of degree, Els{\"a}sser and Tscheuschner \textcite{MaxCutDeg5PLS} showed that binary Boolean \VCSP{s} are \PLS-complete under tight reductions even if the constraints are restricted to \lang{MAXCUT}.
The tight \PLS-reduction means that there are degree $5$ instances where all ascents are exponential -- including the steepest ascent taken by greedy local search.
Monien and Tscheuschner \textcite{MaxCutDeg4} constructed a family of instances of degree $4$ for which they claim all ascents are exponential.
Finally, all ascents are quadratic when the constraint graph has maximum degree $2$ (Theorem 5.6 in Kaznatcheev \textcite{KazThesis}).
Taken together, this only leaves open the question of efficiency of greedy local search for binary Boolean \VCSP{s} of degree $3$.

Treewidth as a sparsity parameter has a more complicated -- and perhaps more interesting -- relationship to the efficiency of local search.
There exists a polynomial time non-local search algorithm for finding not just a local maximum but the global maximum for \VCSP{s} of bounded treewidth~\cite{BB73,VCSPsurvey}.
Thus, binary Boolean \VCSP{s} of bounded treewidth are not hard for \PLS.
Unlike bounded degree, bounded treewidth instances really are a class of simple instances.
But the existence of a polynomial-time non-local search algorithm for solving bounded treewidth \VCSP{s}, does not mean that local search will be efficient at finding even a local peak.
Cohen et al. \textcite{tw7} have already shown that greedy local search requires an exponential number of steps for Boolean \VCSP{s} of pathwidth $7$. 
The simplest set of instances on which greedy local search fails was subsequently lowered to pathwidth $4$~\cite{pw4}.
Recently, Kaznatcheev and Vazquez Alferez \textcite{conditionallySmooth} showed that an old construction of Hopfield networks by Haken and Luby \textcite{hakenSteepest} provides a binary Boolean \VCSP{} of pathwidth $3$ on which greedy local search takes an exponential number of steps.
Since Kaznatcheev, Cohen and Jeavons \textcite{repCP} have shown that \emph{all} local search methods will take at most a quadratic number of steps on tree-structured binary Boolean \VCSP{s} (i.e., treewidth $1$), this leaves only the case of treewidth $2$ as open for the question of whether greedy local search is fast or slow.

There are partial results for these two open questions on degree and treewidth -- mostly focused on \emph{all} or \emph{some} ascents, rather than \emph{steepest} ascents -- but they cannot be further extended.
Poljak \textcite{MaxCutDeg3Easy} showed that if the \VCSP{} instance is allowed only \lang{MAXCUT} constraints then for degree $3$ instances, all ascents are quadratic.
However, Poljak's \textcite{MaxCutDeg3Easy} technique of rewriting degree $3$ \lang{MAXCUT} constrain graphs by instances with small weights cannot extend to arbitrary binary constraints.
There are degree $3$ \VCSP-instances that require exponentially large weights (Example 5.10 in Kaznatcheev \textcite{KazThesis}) and, in their Example 7.2, Kaznatcheev, Cohen and Jeavons \textcite{repCP} gave an explicit family of instances with max degree $3$ where some ascents are exponential.
Similarly, the encouragement-path proof techniques for showing that all local search methods are efficient on \VCSP{s} of treewidth $1$~\cite{repCP} cannot be extended to treewidth $2$.
Specifically, Kaznatcheev, Cohen and Jeavons \textcite{repCP}'s Example 7.2 not only has maximum degree $3$ but also pathwidth $2$.
However, although some ascents are exponentially long in this family of instances, the steepest ascents are all short and greedy local search can find the peak in a linear number of steps.
More generally, Kaznatcheev and van Marle \textcite{pw4} optimistically conjectured that on \emph{any} family of \VCSP-instances of pathwidth $2$, greedy local search will take at most a polynomial number of steps.

In this paper, we resolve these two open questions on the efficiency of greedy local search for \VCSP-instances of degree $3$ and treewidth $2$.\footnote{
In parallel work to ours, van Marle~\cite{pw2MSc} developed a similar construction of exponential steepest ascent.
His construction has the same constraint graph -- with degree $3$ and pathwidth $2$ -- but slightly different weights.
The key difference is a much more involved proof of exponential steepest ascents:
van Marle~\cite{pw2MSc} shows how to simulate a particular prior construction of \emph{some} exponential ascent as a \emph{steepest} ascent on a ``padded'' instance.
This is similar to the prior approaches taken by Cohen et al.~\cite{tw7} and Kazntcheev and van Marle~\cite{pw4}.
In contrast, we focus on how to compose individual gadget's fitness landscapes and their steepest (partial) ascents rather than introducing new subgadgets for simulating whole ascents.
We find our approach simpler, and think that future work can more easily generalize our approach to other local search methods.
}
In \cref{sec:construction}, we construct \VCSP-instances $\mathcal{C}^\pm_{n,\leq n}$ on $6n$ Boolean variables with $7n - 1$ binary constraints and $6n$ unary constraints with a constraint graph of maximum degree $3$ and pathwidth $2$ (\cref{prop:pw2}).
In \cref{sec:ascent}, we show that greedy local search takes $7(2^n - 1)$ steps to solve these simple \VCSP-instances (\cref{thm:steepest-ascent}).
We construct $\mathcal{C}^\pm_{n,\leq m}$ as a path of $m$ gadgets $\mathcal{C}^\pm_{n,m},\mathcal{C}^-_{n,m - 1},\ldots,\mathcal{C}^-_{n,2},\mathcal{C}^-_{n,1}$ 
with consecutive gadgets joined by a single binary constraint.
Our proof of long steepest ascents is by induction on the number of gadgets in the path.
We show that the steepest ascent first flips bits only in the first gadget $\mathcal{C}^\pm_{n,m}$ until it flips the variable that participates in the binary constraint linking $\mathcal{C}^\pm_{n,m}$ to $\mathcal{C}^-_{n,m-1}$.
When this linking variable flips to $1$ this gives us the inductive hypothesis of $\mathcal{C}^+_{n,\leq m - 1}$ and when it flips to $0$, it gives the inductive hypothesis of $\mathcal{C}^-_{n,\leq m - 1}$.
The constraint weights in the gadget are chosen in such a way that after this linking variable is flipped, the next potential flip in $\mathcal{C}^\pm_{n,m}$ increases fitness by a lower amount than the flips in $\mathcal{C}^\pm_{n,\leq m - 1}$.
Thus the steepest ascent continues in the part of the \VCSP{} corresponding to the inductive hypothesis for $7(2^{m - 1} - 1)$ steps, until it reaches the local peak of the sub-instance and finally allows flips to continue in $\mathcal{C}^\pm_{n,m}$ that eventually result in the linking variable flipping back, repeating the recursive process for a second time.

We also show that our \VCSP-instances are what Kaznatcheev and Vazquez Alferez \textcite{conditionallySmooth} defined as `oriented' (\cref{prop:oriented}).
This means that the exponential running time of greedy local search on our instances stands in stark contrast to many other non-greedy local search methods that Kaznatcheev and Vazquez Alferez \textcite{conditionallySmooth} have shown to take a quadratic or fewer steps to solve oriented \VCSP{s}.
From this we conclude that among local search methods, greed is particularly slow.

\section{Background}\label{sec:background}

Given some finite set $V$ of variable indexes with size $|V| = d$,\footnote{Most often this is $V = [d]$, but in our construction of $\mathcal{C}^\pm_{n,\leq m}$ in \cref{sec:construction} it will be $V = [m]\times[6]$.}
an \emph{assignment} is a $d$-dimensional Boolean vector $x \in \{0,1\}^{d}$.
For every $i \in V$, $x_i$ refers to the $i$th entry of $x$.
To refer to a substring with variable indexes $S \subseteq V$, we write the partial assignment $x[S] \in \{0,1\}^S$.
To complete a partial assignment, we write $x\in\{0,1\}^{S}y[V \setminus S]$ to mean that $x_i=y_i$ for $i \in V \setminus S$ and free otherwise.
When assignments are sufficiently short we give $x$ and $y$ explicitly, 
for example, $x_10x_3$ refers to an assignment to $3$ variables where the second variable is set to $0$.
If we want to change just the assignment $x_i$ at index $i$ to a value $b \in \{0,1\}$ then we will write $x[i:b]$.
Two assignments are \emph{adjacent} if they differ on a single variable.
Or more formally: $x,y \in \{0,1\}^{d}$ are adjacent if there exists an index $i \in [d]$ such that $y=x[i:\overline{x_i}]$, where $\overline{x_i}=1-x_i$.

A \emph{fitness landscape} is any pseudo-Boolean function $f: \{0,1\}^V \rightarrow \mathbb{Z}$ together with the above notion of adjacent assignments.
Given any $S \subseteq V$, the sublandscape on $S$ given background $y \in \{0,1\}^{V \setminus S}$ is the function $f$ restricted to inputs $x\in\{0,1\}^{S}y[V \setminus S]$.
An assignment $x$ is called a \emph{local peak} in a fitness landscape if $f(x)\geq f(y)$ for all $y$ adjacent to $x$.
A sequence of assignments $x^0,x^1,\ldots,x^T$ is called an \emph{ascent} in the fitness landscape if $x^{t - 1}$ and $x^t$ are adjacent for all $t\in[T]$,  $f(x^{t - 1})<f(x^t)$, and $x^T$ is a local peak.
An ascent is called a \emph{steepest} ascent, if for all $t\in[T]$ and all assignments $y$ adjacent to $x^{t - 1}$ it  is the case that $f(x^t)\geq f(y)$.
Any local search method (that only takes increasing steps) follows an ascent.
Greedy local search follows a steepest ascent.

A binary Boolean valued constraint satisfaction problem (\VCSP) on $d$ Boolean variables with indexes in $V$ is a set of constraints $\mathcal{C} = \{c_S\}$. 
Each constraint is a weight $c_S \in \mathbb{Z} \setminus \{0\}$ with an associated \emph{scope} $S \subseteq V$ of size $|S| \leq 2$.
We say that $c_S$ is a \emph{unary} constraint if $|S| = 1$, and a \emph{binary} constraint if $|S| = 2$.
Overloading notation, the set of constraint $\mathcal{C}$ \emph{implements} a pseudo-Boolean function $\mathcal{C} :\{0,1\}^d\to \mathbb{Z}$ that we call the \emph{fitness function}: 
\begin{equation}
\mathcal{C}(x) = 
c_{\emptyset}+\sum_{c_i \in \mathcal{C}} c_ix_i+\sum_{c_{\{i,j\}} \in \mathcal{C}}c_{\{i,j\}}x_ix_j
\end{equation}
Solving $\mathcal{C}$ means finding an assignment $x$, that maximizes the fitness function $\mathcal{C}(x)$.
\footnote{
One could also consider a \VCSP{}-instance as a set of constraints $\hat{\mathcal{C}} = \{ C_S \}$ where each $C_S: \{0,1\}^S \rightarrow \mathbb{Z}$ is a binary function. 
This formulation is equivalent to our formulation above because an arbitrary constraint with scope $S$ can be expressed as a polynomial. For example, take $S=\{i,j\}$:
\begin{align*}
    C_{S}(x_i,x_j) & = C_{S}(0,0) (1 - x_i)(1- x_j) + C_{S}(1,0)x_i(1- x_j) + C_{S}(0,1)(1 - x_i)x_j + C_{S}(1,1)x_ix_j \\
    & = C_S(0,0) + (C_S(1,0) - C_S(0,0))x_i + (C_S(0,1) - C_S(0,0))x_j \\
    & \quad\quad\quad\quad\quad + (C_S(1,1) - C_S(0,1) - C_S(1,0) + C_S(0,0))x_ix_j
\end{align*}
The second equality groups alike monomials. 
One can convert from the $\hat{\mathcal{C}}$ formulation to the $\mathcal{C}$ formulation by summing alike monomial terms across all the constraints. That is,
the constant terms $C_S(0,0)$ are aggregated into $c_\emptyset$, all the $C_{\{i,\_\}}$ (i.e., coefficients appearing before $x_i$) aggregate into $c_i$, and $c_{\{i,j\}} = C_{\{i,j\}}(1,1) - C_{\{i,j\}}(0,1) - C_{\{i,j\}}(1,0) + C_{\{i,j\}}(0,0)$ (i.e., coefficient appearing before $x_ix_j$).
It takes linear time to covert from $\hat{\mathcal{C}}$ to $\mathcal{C}$ (see Theorem 3.4 in Kaznatcheev, Cohen and Jeavons \textcite{repCP}).
}

When discussing graph-theoretical properties of $\mathcal{C}$, we treat the scopes of binary constraints as edges.
So $V(\mathcal{C}) = V$ and $E(\mathcal{C}) = \{{i,j} \; | \; i \neq j \text{ and } c_{\{i,j\}} \in \mathcal{C}\}$ is the set of scopes of the binary constraints of $\mathcal{C}$.
For each variable index $i \in V$, we define the neighbourhood $N_\mathcal{C}(i) = \{ j \; | \; {i,j} \in E(\mathcal{C})\}$ as the set of variable indexes in $V \setminus \{i\}$ that appear in a constraint with $i$.
In this paper, we measure the simplicity of a \VCSP-instance $\mathcal{C}$ by the maximum degree and by the pathwidth of its constraint graph.
Given a graph $G=(V,E)$, the \emph{pathwidth} of $G$ is the minimum possible width of a path decomposition of $G$.
A \emph{path decomposition} of $G$ is a sequence of sets $P=\{X_1,X_2,\ldots, X_p\}$ where $X_r\subset V(G)$ for $r\in[p]$ with the following three properties:
\begin{enumerate}
    \item Every vertex $v\in V(G)$ is in at least one set $X_r$.
    \item For every edge $\{u,v\}\in E(G)$ there exists an $r \in [p]$ such that $X_r$ contains both $u$ and $v$.
    \item For every vertex $v\in V(G)$, if $v\in X_r\cap X_s$ %for some $r\leq s$ 
    then $u\in X_\ell$ for all $\ell$ such that $r\leq \ell\leq s$.
\end{enumerate}
The \emph{width} of a path decomposition is defined as  $\max_{X_r\in P}|X_r|-1$.
We refer the reader to \cite{parameterizedBook} for more details and for the definition of treewidth, and to \cite{graphBook} for standard graph terminology. 
%It is known that graphs of pathwidth $1$ cannot contain $K_3$ as a minor\cite{obstructionsPW1}; and that the pathwidth $\leq$ treewidth of a graph.

It is often useful to see how the value of the fitness function $\mathcal{C}$ changes when a single variable is modified.
In particular, we denote with $\nabla_i\mathcal{C}(x)=\mathcal{C}(x[i:1])-\mathcal{C}(x[i:0])$ the fitness change associated with changing variable $x_i=0$ to $x_i=1$ given some background assignment $x$.
It is easy to see that $\nabla_i\mathcal{C}(x)=c_i+\sum_{j\in N_\mathcal{C}(i)}c_{\{i,j\}}x_j$, and the value of $\nabla_i\mathcal{C}(x)$ depends only on the assignment to variables with indexes in $N_\mathcal{C}(i)$.
We use this to overload $\nabla_i\mathcal{C}$ to partial assignments: if $y\in\{0,1\}^{N(i)}$ we consider $\nabla_i\mathcal{C}(y)$ to be well defined. 

We say that $x_i$ has \emph{preferred assignment} $1$ in background $x$ if $\nabla_i\mathcal{C}(x)>0$ and preferred assignment $0$ in background $x$ if $\nabla_i\mathcal{C}(x)<0$. 

\begin{definition}\label{def:sign-depends}
    Given two indexes $i\neq j$ we say that $i$ \emph{sign-depends} on $j$ in background assignment $x$ when $\text{sign}(\nabla_i\mathcal{C}(x))\neq\text{sign}(\nabla_i\mathcal{C}(x[j:\overline{x_j}]))$.
    If there is no background assignment $x$ such that $i$ sign-depends on $j$ then we say that $i$ does not sign depend on $j$.
    If for all $j\neq i$ we have that $i$ does not sign-depend on $j$ then we say that $i$ is sign independent.
\end{definition}
In other words, \cref{def:sign-depends} tells us that if $i$ sign-depends on $j$, the preferred assignment of variable $x_i$ depends on the assignment to variable $x_j$.

\begin{definition}[Kaznatcheev and Vazquez Alferez \textcite{conditionallySmooth}]\label{def:oriented}
    A \VCSP{-}instance $\mathcal{C}$ is \emph{oriented} if for every pair of indexes $\{i,j\}$ we have that either $i$ does not sign-depend on $j$ or $j$ does not sign depend on $i$.
    If a \VCSP{-}instance $\mathcal{C}$ is oriented and $j$ sign depends on $i$ we assign a direction from $i$ to $j$ to the edge $\{i,j\}\in E(\mathcal{C})$ (i.e. we orient the edge, hence the name).
\end{definition}
The constraint graphs of oriented \VCSP{s} have no directed cycles \cite{conditionallySmooth}.
Thus, if $y$ is any assignment to the first $k$ variables of a topological ordering of the variables of an oriented constraint graph, then $k$ is sign-independent given background $y$. 
In other words, if $\mathcal{C}$ is oriented, there is an ordering on the variables such that later variables are conditionally-independent of earlier ones. 
This implies that the fitness landscape of $\mathcal{C}$ is single peaked on every sublandscape~\cite{conditionallySmooth}.
Depending on the research community, such landscapes are known as semismooth fitness landscapes~\cite{evoPLS,conditionallySmooth}, completely unimodal pseudo-Boolean functions~\cite{completelyUnimodal}, or acyclic unique-sink orientations of the hypercube~\cite{AUSO_Thesis,AUSO}.
Given that fitness landscapes implemented by oriented \VCSP{s} are single-peaked, we use $x^*(\mathcal{C})$ to denote the peak of the landscape implemented by the oriented \VCSP{-}instance $\mathcal{C}$.

\section{Construction of $\mathcal{C}^\pm_{n,\leq m}$}
\label{sec:construction}

Given two parameters $1 \leq m \leq n$, in this section, we construct the \VCSP{-}instances $\mathcal{C}^+_{n,\leq m}$ and $\mathcal{C}^-_{n,\leq m}$ on $6m$ variables and, in \cref{sec:ascent}, show that they both have a steepest ascent of length $7(2^m - 1)$. %(Theorem~\ref{thm:steepest-ascent}).
We construct the $\mathcal{C}^\pm_{n,\leq m}$ as a path of $m$ gadgets $\mathcal{C}^\pm_{n,m},\mathcal{C}^-_{n,m - 1},\ldots,\mathcal{C}^-_{n,2},\mathcal{C}^-_{n,1}$
where each gadget $\mathcal{C}^\pm_{n,k}$ is defined on $6$ variables $V_k = \{(k,i) \; | \; 1 \leq i \leq 6\}$.
For notational convenience, we define $V_{\leq m} := \cup_{k = 1}^m V_k$.
Note that %$\mathcal{C}^+_{n,\leq m}$ and $\mathcal{C}^-_{n,\leq m}$ 
the two \VCSP{s} are exactly the same except on the $m$th gadget where $\mathcal{C}^+_{n,\leq m}$ has $\mathcal{C}^+_{n,m}$ and $\mathcal{C}^-_{n,\leq m}$ has $\mathcal{C}^-_{n,m}$.
We will present the construction of the gadgets $\mathcal{C}^\pm_{n,k}$ in two stages.
First, we will define the scopes of all the constraints to get the constraint graph
and show that these constraint graphs are very sparse (\cref{prop:pw2}).
Second, we will assign weights to the constraints and show that the \VCSP{s} are oriented (\cref{prop:oriented}).

\input{img-and-andn}

Both the $\mathcal{C}^-_{n,k}$ and $\mathcal{C}^+_{n,k}$ gadgets have all six unary constraints and the same six binary constraints with scopes $\{(k,1),(k,2)\}$, $\{(k,2),(k,3)\}$, $\{(k,3),(k,6)\}$, $\{(k,1),(k,4)\}$, $\{(k,4),(k,5)\}$, $\{(k,5),(k,6)\}$.
Finally, we connect adjacent gadgets with a single binary constraint with scope $\{(k,6),(k-1,1)\}$.
The constraint graph of $\mathcal{C}^-_{n,k}$ along with the connections to the adjacent gadgets at $k+1$ and $k-1$ are shown in \cref{fig:steepestGadget}.
It is not hard to check based on the above definition that the constraint graphs of $\mathcal{C}^\pm_{n,\leq m}$ are sparse:
\begin{proposition}
The constraint graph of $\mathcal{C}^\pm_{n,\leq m}$ has maximum degree $3$ and pathwidth $2$.
\label{prop:pw2}
\end{proposition}

\begin{proof}
The constraint graph of each $\mathcal{C}^\pm_{n,k}$ is a cycle so every vertex has degree $2$.
To create $\mathcal{C}^\pm_{n,\leq m}$ we add a single edge between consecutive gadgets, this raises the maximum degree to $3$.
For the pathwidth:

\noindent ($\Rightarrow$) Path decomposition for $\mathcal{C}^\pm_{n,k}$ and adjacent variables:
$\{(k+1,6),(k,1)\}$, $\{(k,1),(k,2),(k,4)\}$, $\{(k,2),(k,3),(k,4)\}$, $\{(k,3),(k,4),(k,5)\}$, $\{(k,3),(k,5),(k,6)\}$, $\{(k,6),(k-1,1)\}$.

\noindent ($\Leftarrow$) Contracting $\{(k,1),(k,2)\}$, $\{(k,3),(k,6)\}$ and $\{(k,4),(k,5)\}$ shows $K_3$ is a minor of $\mathcal{C}^\pm_{n,k}$.
\end{proof}

We define the weights of the constraints sequentially using parameters $M_k = 6(2^k - 2)$, $S = 2n + 1$, and $s_k = n + 1 - k$.
Since $C^-_{n,k}$ and $C^+_{n,k}$ are the same except for the unary on $(k,1)$, we use $c$ for all the weights except for $c^-_{(k,1)}$ (\cref{eq:cminus}) and $c^+_{(k,1)}$ (\cref{eq:cplus}):

\begin{align}
c_{\{(k,6),(k-1,1)\}} & & = \quad& M_k S \\
c_{(k,6)} & = - (|c_{\{(k,6),(k-1,1)\}}| + S) & = \quad& -(M_k + 1)S \\
c_{\{(k,3),(k,6)\}} & = |c_{(k,6)}| + S & = \quad& (M_k + 2)S \\
c_{\{(k,5),(k,6)\}} & = -|c_{\{(k,3),(k,6)\}}| & = \quad& -(M_k + 2)S \\
c_{(k,3)} & = -(|c_{\{(k,3),(k,6)\}}| + S) & = \quad & -(M_k + 3)S \\
c_{\{(k,2),(k,3)\}} & = |c_{(k,3)}| + S & = \quad & (M_k + 4)S \\
c_{(k,2)} & = - (|c_{\{(k,2),(k,3)\}}| + s_k) & = \quad & -(M_k + 4)S - s_k \label{eq:smallStep1} \\
c_{\{(k,1),(k,2)\}} & = |c_{(k,2)}| + S - s_k & = \quad & (M_k + 5)S \\
c_{(k,5)} & = &  & - S \\
c_{\{(k,4),(k,5)\}} & = |c_{(k,5)} + c_{\{(k,5),(k,6)\}}| + S & = \quad & (M_k + 4)S \\
c_{(k,4)} & = -(|c_{\{(k,4),(k,5)\}}| + S) & = \quad & -(M_k + 5)S \\
c_{\{(k,1),(k,4)\}} & = |c_{(k,4)}| + s_k & = \quad & (M_k + 5)S + s_k \label{eq:smallStep2} \\
c^-_{(k,1)} & = -(|c_{\{(k,1),(k,2)\}} + c_{\{(k,1),(k,4)\}}| + S - s_k) & = \quad & -(2(M_k + 5) + 1)S \label{eq:cminus}\\
c_{\{(k + 1,6),(k,1)\}} & = |c_{(k,1)}| + S & = \quad & \underbrace{2(M_k + 6)}_{M_{k + 1}}S \\
c^+_{(k,1)} & = c^-_{(k,1)} + c_{\{(k + 1,6),(k,1)\}}  & = \quad & S \label{eq:cplus}
\end{align}

\noindent  The above weights are shown on the constraint graph in \cref{fig:steepestGadget}.
This structure of weights has three important features that ensure the \VCSP{} is oriented, that let us recurse on smaller $k$, and that help us control the steepest ascent.
We have set the above weights in such a way that the following three properties hold in $\mathcal{C}^-_{n,k}$: 
\begin{enumerate}[(a)]
\item all the unaries are negative, \label{point:all_negative} %(except for $c^+_{(k,1)} = S > 0$ which occurs only when there is no constraint with scope $\{(k + 1,6),(k,1)\}$),
\item the magnitude of each variable's unary is greater than the binary constraint from that variable to a variable in the gadget with higher second index (or than the sum of the two binaries in the case of $|c^-_{(k,1)}| > c_{\{(k,1),(k,2)\}} + c_{\{(k,1),(k,4)\}}$)\label{point:outgoing}, and
\item the sum of the weights of any subset of `incoming' binaries on a variable (i.e., binary constraint to that variable from a variable in the gadget with lower second index) is either non-positive or greater than the magnitude of the unary (or the sum of the unary and any negative `outgoing' binaries in the case of $| c_{\{(k,4),(k,5)\}}| > |c_{(k,5)}|+ |c_{\{(k,5),(k,6)\}}|$).\label{point:incomming}
\end{enumerate}
These three properties ensure that the resulting \VCSP{} is oriented:

\begin{proposition}
    $\mathcal{C}^\pm_{n,k}$ and $\mathcal{C}^\pm_{n,\leq m}$ are oriented.
    \label{prop:oriented}
\end{proposition}
\begin{table}[tbp]
    \begin{center}
        \begin{tabular}{|l|l|l|c||c|c|c|c|}
             \hline
              \multicolumn{4}{ |c |}{}  &\multicolumn{4}{ c | }{( $x_{(k,i)}$ , $x_{(k,j)}$ )}\\
             \hline
             $h$ & $i$ & $j$ & $\nabla_{(k,h)}^\pm$ & $(0,0)$ & $(0,1)$ & $(1,0)$ & $(1,1)$\\
             \hline
             \hline
             \multirow{ 2}{*}{$1$} &\multirow{ 2}{*}{$4$} & \multirow{ 2}{*}{$2$} & $\nabla_{(k,1)}^+$& \cellcolor{posCol} $\scriptstyle S$ & \cellcolor{posCol} $\scriptstyle (M_k+6)S+s_k$ & \cellcolor{posCol} $\scriptstyle (M_k+6)S$ & \cellcolor{posCol} $\scriptstyle (2M_k+11)S+s_k$\\
             \cline{4-4}
             
              & & &$\nabla_{(k,1)}^-$ & \cellcolor{negCol} $\scriptstyle -(2(M_k+5)+1)S$ & \cellcolor{negCol} $\scriptstyle -(M_k+6)S+s_k$ & \cellcolor{negCol} $\scriptstyle -(M_k+6)S$ & \cellcolor{negCol} $\scriptstyle -S+s_k$\\
             \hline
             
             $2$ & $1$ & $3$ &$\nabla_{(k,2)}^\pm$ & \cellcolor{negCol} $\scriptstyle -(M_k+4)S-s_k$ & \cellcolor{negCol} $\scriptstyle -s_k$ & \cellcolor{posCol} $\scriptstyle S-s_k$ & \cellcolor{posCol} $\scriptstyle (M_k+5)S-s_k$ \\
             \hline
             
             $3$ & $2$ & $6$& $\nabla_{(k,3)}^\pm$ & \cellcolor{negCol} $\scriptstyle -(M_k+3)S$ & \cellcolor{negCol} $\scriptstyle -S$ & \cellcolor{posCol} $\scriptstyle S$ & \cellcolor{posCol}$\scriptstyle (M_k+3)S$ \\
             \hline
             
             $4$ & $1$ & $5$ & $\nabla_{(k,4)}^\pm$ & \cellcolor{negCol} $\scriptstyle -(M_k+5)S$ & \cellcolor{negCol} $\scriptstyle -S$ & \cellcolor{posCol} $\scriptstyle s_k$ & \cellcolor{posCol} $\scriptstyle (M_k+4)S+s_k$\\
             \hline
             
             $5$ & $4$ & $6$ &$\nabla_{(k,5)}^\pm$ & \cellcolor{negCol} $\scriptstyle -S$ & \cellcolor{negCol} $\scriptstyle -(M_k+3)S$ & \cellcolor{posCol} $\scriptstyle (M_k+3)S$ & \cellcolor{posCol} $\scriptstyle S$\\
             \hline
             
             $6$ & $3$ & $5$ &$\nabla_{(k,6)}^\pm$ & \cellcolor{negCol} $\scriptstyle -(M_k+1)S$ & \cellcolor{posCol} $\scriptstyle S$ & \cellcolor{negCol} $\scriptstyle -(2M_k+3)S$ & \cellcolor{negCol} $\scriptstyle -(M_k+1)S$\\
             \hline
             
        \end{tabular}
    \end{center}
    \caption{Fitness change $\nabla_{(k,h)}\mathcal{C}^\pm_{n,k}$ incurred by flipping the assignment of variable $(k,h)$ from $x_{(k,h)}=0$ to $x_{(k,h)}=1$ given the assignments of its two neighbors $x_{(k,i)}$ and $x_{(k,j)}$ in $\mathcal{C}^\pm_{n,k}$. We abbreviate $\nabla_{(k,h)}\mathcal{C}^\pm_{n,k}$ as $\nabla_{(k,h)}^\pm$. %The white cells highlight the positive changes, whilst 
    The dark gray cells highlight negative changes, and light gray cells highlight positive changes. 
    }
    \label{tab:oriented}
\end{table}

\begin{proof} 
    First we prove that $\mathcal{C}^\pm_{n,k}$ are oriented as in \cref{fig:steepestGadget}.
    Note that $\mathcal{C}^\pm_{n,k}$ are cycles and every vertex has degree $2$.
    We will denote by $(k,i)$ and $(k,j)$ the two neighbors of an arbitrary variable $(k,h)\in V_k$. 
    \cref{tab:oriented} shows the fitness change $\nabla_{(k,h)}^\pm\mathcal{C}(x_{(k,i)}x_{(k,j)})$ for all possible assignments to $x_{(k,i)}$ and $x_{(k,j)}$.
    The cells of \cref{tab:oriented} are colored according to the sign of $\nabla_{(k,h)}^\pm\mathcal{C}$.
    It suffices to check \cref{tab:oriented} against \cref{def:oriented}:
    The sign of $\nabla_{(k,1)}^+$ is always positive, and the sign of $\nabla_{(k,1)}^-$ is always negative, regardless of the assignment to $(k,2)$ and $(k,4)$, so $(k,1)$ does not sign depend on either $(k,2)$ or $(k,4)$ in $\nabla_{(k,1)}^\pm$. 
    
    On the other hand, for the rows where $h=2,3,4$ and $5$ the two columns where $x_{(k,i)}=0$ are negative whilst the two columns corresponding to $x_{(k,i)}=1$ are positive.
    This means that $(k,2),(k,3),(k,4)$ and $(k,5)$ sign-depend, respectively, on $(k,1),(k,2),(k,1)$ and $(k,4)$; but do not sign-depend, respectively, on $(k,3),(k,6),(k,5)$ and $(k,6)$.
    Finally, from the fact that the row with $h=6$ has different signs on the two columns where $x_{(k,3)}=0$, and also different signs on the two columns where $x_{(k,5)}=1$ we can see that $(k,6)$ sign-depends on $(k,3)$ and $(k,5)$ in $\nabla_{(k,1)}^\pm$.
    
    Second, to show that $\mathcal{C}^\pm_{n,\leq m}$ are oriented, it suffices to show that for $1\leq k<m$ the preferred assignment to $x_{(k+1,6)}$ is independent of the assignment to $x_{(k,1)}$.
    This is equivalent to showing that the sign of the last row in \cref{tab:oriented} does not change if we add $M_kS$ to each column, which is obviously true.
    This is sufficient to show the statement of the proposition.
    
    To show the additional property that the orientation of the edge $\{x_{(k+1,6)},x_{(k,1)}\}$ is as it appears in \cref{fig:steepestGadget}, we must show that $(k,1)$ sign-depends on the assignment to $(k+1,6)$.
    But this can be seen from the fact that the first row of our table is equivalent to having $x_{(k+1,6)}=1$  in the background, and the second row is equivalent to having $x_{(k+1,6)}=0$.
    Since the first and second rows have different signs the sign-dependence follows. 
\end{proof}

\noindent The weight of $c^+_{(k,1)}$ in  \cref{eq:cplus} is set so that we have the next two properties: 
\begin{enumerate}[(a)]
\setcounter{enumi}{3}
\item If we fix $x_{(k,6)} = 0$ then the sublandscape spanned by $V_{\leq k - 1}$ is the same as the landscape implemented by $\mathcal{C}^-_{n,\leq k - 1}$, and \label{point:zero_out}
\item if we fix $x_{(k,6)} = 1$ then the sublandscape spanned by $V_{\leq k - 1}$ is the same as the landscape implemented by $\mathcal{C}^+_{n,\leq k - 1}$. \label{point:one_out}
\end{enumerate}
This allows us to analyze ascents on $\mathcal{C}^\pm_{n,\leq m}$ inductively because when $x_{(m,6)}$ is fixed to $1$ or $0$ we can use the analysis from our inductive hypothesis of $\mathcal{C}^+_{n,\leq m - 1}$ or $\mathcal{C}^-_{n,\leq m - 1}$, respectively.
It also makes it very easy to check for the peaks of our landscapes:

\begin{proposition}
The semismooth fitness landscapes of $\mathcal{C}^-_{n,k}$, $\mathcal{C}^-_{n,\leq m}$, $\mathcal{C}^+_{n,k}$, and $\mathcal{C}^+_{n,\leq m}$  have their unique fitness peaks at $x^*(\mathcal{C}^-_{n,k}) = 000000$, $x^*(\mathcal{C}^-_{n,\leq m}) = 0^{6m}$, $x^*(\mathcal{C}^+_{n,k}) = 111110$, and  $x^*(\mathcal{C}^+_{n,\leq m}) = x^*(\mathcal{C}^+_{n,m})0^{6(m - 1)} =  1111100^{6(m - 1)}$, respectively.\footnote{
For convenience, when we specify an assignment $x$ the variables are ordered from left to right by decreasing first index and, to break ties, by increasing second index. %: $x_{(m,1)}x_{(m,2)}\ldots x_{(m,6)}x_{(m-1,1)}x_{(m-1,2)}\ldots x_{(1,5)}x_{(1,6)}$.
For example, for $\mathcal{C}^+_{n,\leq m}$ the assignment $x=0100010^{6(m-1)}$ sets $x_{(m,2)}=1,x_{(m,6)}=1$ and all other variables to $0$.
}
\end{proposition}

\begin{proof}
By property~(\ref{point:all_negative}), all the unaries in $\mathcal{C}^-_{n,k}$ and $\mathcal{C}^-_{n,\leq m}$ are negative, so the all zero assignment is a local peak.
Since the landscapes are semismooth from \cref{prop:oriented}, this is also the unique global peak.

For $\mathcal{C}^+_{n,k}$, $x_{(k,1)}$ is conditionally-independent of all other variables and has $c^+_{(k,1)} = S > 0$, so the preferred assignment is $x_{(k,1)} = 1$.
With $x_{(k,1)}$ set to $1$, the preferred assignments for $x_{(k,2)}$ and $x_{(k,4)}$ are also $1$; and based on those, the preferred assignments for $x_{(k,3)}$ and $x_{(k,5)}$ are also $1$.
This leaves $x_{(k,6)}$ which, conditional on $x_{(k,3)} = x_{(k,5)} = 1$, prefers $x_{(k,6)} = 0$.
Overall, this gives $x^*(\mathcal{C}^+_{n,k}) = 111110$.

Since $x_{(k,6)} = 0$, property~(\ref{point:zero_out}) implies that $x^*(\mathcal{C}^+_{n,\leq m}) = x^*(\mathcal{C}^+_{n,m})x^*(\mathcal{C}^-_{n,\leq m - 1})$. % = 1111100^{6(m - 1)}$.
\end{proof}

\noindent Finally, the difference in increase of \cref{eq:smallStep1,eq:smallStep2} from the other weights ensures that:
\begin{enumerate}[(a)]
\setcounter{enumi}{5}
\item steps from assignments $0\underbar{1}1x_4x_5x_6$ to $0\boldsymbol{0}1x_4x_5x_6$ and from $1x_2x_3\underline{0}0x_6$ to $1x_2x_3\boldsymbol{1}0x_6$ increase fitness by exactly $s_k \leq n$, and 
\item all other fitness increasing steps increase fitness by at least $S - s_k > n$.
\end{enumerate}
As we will see in the next section, these ``small'' steps allow us to control in what order each block $V_k$ of variables appears in the steps of the steepest ascent.
Combined with properties~(\ref{point:zero_out}) and (\ref{point:one_out}) this lets us show the exponential steepest ascent by induction on $m$.

\section{Exponential steepest ascent in the landscape of $\mathcal{C}^\pm_{n,\leq m}$}
\label{sec:ascent}

We can now show that it takes a large number of steps to go from the assignment $x^*(\mathcal{C}_{n,\leq m}^-) = 0^{6m}$ to the peak $x^*(\mathcal{C}_{n,\leq m}^+) = 1111100^{6(m - 1)}$ in the semismooth fitness landscape implemented by $\mathcal{C}_{n,\leq m}^+$ (and vice versa for the landscape implemented by $\mathcal{C}_{n,\leq m}^-$):

\begin{theorem}\label{thm:steepest-ascent}
    Both the steepest ascents starting 
    %\begin{enumerate}[(a)] 
    %\item 
    from $x^*(\mathcal{C}_{n,\leq m}^-) = 0^{6m}$ in the fitness landscape of $\mathcal{C}_{n,\leq m}^+$ and
    %\item 
    from $x^*(\mathcal{C}_{n,\leq m}^+) = 111110 0^{6(m-1)}$ in the fitness landscape of $\mathcal{C}_{n,\leq m}^-$  
    %\end{enumerate}
    have length $7(2^m-1)$, where each step increases fitness by at least $s_m$.
\end{theorem}

\begin{figure}
    \centering
    \begin{subfigure}[b]{0.49\textwidth}
    \include{img-landscape-zeros}
    \caption{Partial fitness landscape of $\mathcal{C}^+_{n,k}$}
    \label{fig:landscape-zeros}
    \end{subfigure} \hfill \begin{subfigure}[b]{0.49\textwidth}
    \include{img-landscape-ones}
    \caption{Partial fitness landscape of $\mathcal{C}^-_{n,k}$}
    \label{fig:landscape-ones}
    \end{subfigure}
    \caption{All ascents from $000000$ in the fitness landscape of $\mathcal{C}^+_{n,k}$ \textbf{(a)} and from $111110$ in the fitness landscape of $\mathcal{C}^-_{n,k}$ \textbf{(b)}. 
    The increase in fitness of each step is shown on the edges, with small increments ($s_k$) highlighted in red.
    The steepest ascent is shown in bold. 
    For emphasis, bits that lead to a fitness increase when flipped are underlined, and the bit flipped by steepest ascent is bolded and numbered in the same way as in \cref{eq:firstSteps,eq:firstRec,eq:secondSteps,eq:secondRec} and \cref{eq:firstSteps_minus,eq:firstRec_minus,eq:secondSteps_minus,eq:secondRec_minus}.}
    \label{fig:landscapes}
\end{figure}

\begin{proof}
Our proof is by induction on the number of gadgets $m$.
We will show that by adding the gadget $\mathcal{C}^\pm_{n,m}$, steepest ascents of length $T_{m - 1}$ in the landscape of $\mathcal{C}^\pm_{n,\leq m - 1}$ will convert to steepest ascents of length $T_m = 7 + 2T_{m - 1}$ in the landscape of $\mathcal{C}^\pm_{n,\leq m}$.
To do this, we will look at all the ascents that take us from $x^*(\mathcal{C}^-_{n,m}) = 000000$ to $x^*(\mathcal{C}^+_{n,m}) = 111110$ in the landscape of the gadget $\mathcal{C}^+_{n,m}$ and vice-versa for the gadget $\mathcal{C}^+_{n,m}$.
All of these ascents are in \cref{fig:landscape-zeros} for $\mathcal{C}^+_{n,m}$ and in \cref{fig:landscape-ones} for $\mathcal{C}^-_{n,m}$.
Each arrow in \cref{fig:landscapes} is labeled by the fitness increase from flipping the appropriate bit.
The steepest ascent is bolded.

As we can see from the figures,
although there exists a minimal ascent of length five between these two assignments that does not flip the $x_{(m,6)}$ variable, this is not the steepest ascent.
Instead, the steepest ascent from $x^*(\mathcal{C}_{n,m}^-)$ to $x^*(\mathcal{C}_{n,m}^+)$ in the fitness landscape of $\mathcal{C}^+_{n,m}$ (and vice-versa for $\mathcal{C}^-_{n,m}$) takes seven steps and flips $x_{(m,6)}$ twice (at step \Circled{4} and \Circled{7}).

This double flip of $x_{(m,6)}$ is what creates the recursion in $\mathcal{C}^\pm_{n,\leq m}$ that forces the steepest ascent in $\mathcal{C}^\pm_m$ to trigger twice as many ascents in $\mathcal{C}^\pm_{m-1}$.
Specifically, although the first four steps of the steepest ascent in the landscape of $\mathcal{C}^\pm_{n,m}$ increase fitness by a large amount $\geq S - s_m > n$, step \Circled{5} increases fitness by only $s_m$.
Thus, the steepest ascent in the sublandscape spanned by $V_m$ `pauses' after step \Circled{4} and lets the steepest ascents in $V_{m - 1}$ take over with steps that increase fitness by an amount $\geq s_{m - 1} > s_m$.

With this intuition in mind, look at the steepest ascent in the fitness landscape of $\mathcal{C}^+_{n,\leq m}$ starting from $0^{6m}$.
For brevity, define the (partial) assignments $x^*_\pm$ as the peaks of $\mathcal{C}^\pm_{n,\leq m - 1}$: 
\begin{align}
x^*_- & := x^*(\mathcal{C}^-_{n,\leq m - 1}) = 0^{6(m-1)}\text{, and} \\ 
x^*_+ & := x^*(\mathcal{C}^+_{n,\leq m - 1}) = 1111100^{6(m - 2)}.
\end{align}
Now we can rewrite our starting assignment $x^*(\mathcal{C}^-_{n,\leq m}) = 0^{6m}$ as $000000 x^*_-$ %to highlight that the $m$th gadget is the only one that has bits that are not at equilibrium.
and note that the first four flips are entirely in $V_m$:
\begin{equation}
\overbrace{\underbar{\textbf 0}00000x^*_-}^{x^*(\mathcal{C}^-_{n,\leq m})} \xrightarrow{\Circled{1}} 
1\underbar{\textbf 0}0\underbar{0}00x^*_- \xrightarrow{\Circled{2}} 
11\underbar{\textbf 0}\underbar{0}00x^*_- \xrightarrow{\Circled{3}} 
111\underbar{0}0\underbar{\textbf 0}x^*_- \xrightarrow{\Circled{4}} 
111\underbar{0}01\underline{\textbf x^*_-}
\label{eq:firstSteps}
\end{equation}
where the variables that can flip are underlined and the variable that will be flipped by steepest ascent is bolded.
For step \Circled{1} there is only one choice to flip. 
On the subsequent two steps (\Circled{2},\Circled{3}) the first of the two $0$s is chosen by steepest ascent because they increase fitness by $S - s_m$ and $S$ (which are both $> n$) while flipping the second $0$ to a $1$ would increase fitness by only $s_m \leq n$.

Step \Circled{4} is the most interesting.
It flips $x_{(m,6)}$ since that increases fitness by $S > n \geq s_m$. 
In so doing, this flip transforms the landscape for variables on $V_{\leq m - 1}$ from the one implemented by $\mathcal{C}^-_{n,\leq m - 1}$ to the one implemented by $\mathcal{C}^+_{n,\leq m - 1}$.
In the landscape of  $\mathcal{C}^+_{n,\leq m - 1}$, $x^*_-$ is no longer the peak, and hence it is underlined.
Furthermore $x^*_-$ is bolded because, by construction, all fitness increasing flips of variables in $V_{\leq m - 1}$ increase fitness by $\geq s_{m - 1} > s_m$ and so steepest ascent will flip variables in $V_{\leq m - 1}$ instead of the $x_{(m,3)}$ flip that only increases fitness by $s_m$:
\begin{equation}
111\underbar{0}01\underline{\textbf x^*_-} \xrightarrowdashed{T_{m-1} \text{ steps of } V_{\leq m - 1} \text{ in landscape of } \mathcal{C}^+_{n,\leq m - 1} }
111\underbar{\textbf 0}01x^*_+
\label{eq:firstRec}
\end{equation}
Once all the steps in $V_{\leq m - 1}$ are taken, steepest ascent can return to $V_m$, where the only remaining fitness-increasing step is the small step at $x_{(m,4)}$ that increases fitness by only $s_m$.
This step subsequently opens two more steps in $V_m$:
\begin{equation}
111\underbar{\textbf 0}01x^*_+ \xrightarrow{\Circled{5}}
1111\underbar{\textbf 0}1x^*_+ \xrightarrow{\Circled{6}}
11111\underbar{\textbf 1}x^*_+ \xrightarrow{\Circled{7}}
111110\underline{\textbf x^*_+}
\label{eq:secondSteps}
\end{equation}
As with the first four steps in \cref{eq:firstSteps}, the most interesting step is the final step (\Circled{7}).
It flips $x_{(m,6)}$ from $1$ to $x_{(m,6)} = 0$.
In so doing, this flip transforms the landscape for variables on $V_{\leq m - 1}$ from the one implemented by $\mathcal{C}^+_{n,\leq m - 1}$ to the one implemented by $\mathcal{C}^-_{n,\leq m - 1}$.
Thus, `undoing' step \Circled{4}.
In the landscape of  $\mathcal{C}^-_{n,\leq m - 1}$, $x^*_+$ is no longer the peak, and hence underlined and bolded.
Steepest ascent finishes with all remaining steps in $V_{\leq m-1}$:
\begin{equation}
111110\underline{\textbf x^*_+} \xrightarrowdashed{T_{m-1} \text{ steps of } V_{\leq m - 1} \text{ in landscape of } \mathcal{C}^-_{n,\leq m - 1} }
\underbrace{111110x^*_-}_{x^*(\mathcal{C}^+_{n,\leq m})}
\label{eq:secondRec}
\end{equation}
Similar to the steepest ascent in \cref{eq:firstSteps,eq:firstRec,eq:secondSteps,eq:secondRec},  the steepest ascent in the fitness landscape of $\mathcal{C}^-_{n,\leq m}$ starting from the assignment $x^*(\mathcal{C}^+_{n,\leq m})$ has the following steps:
\begin{align}
& \overbrace{\underbar{\textbf 1}11110x^*_-}^{x^*(\mathcal{C}^+_{n,\leq m})} \xrightarrow{\Circled{1}} 
0\underbar{1}1\underbar{\textbf 1}10x^*_- \xrightarrow{\Circled{2}} 
0\underbar{1}10\underbar{\textbf 1}0x^*_- \xrightarrow{\Circled{3}} 
0\underbar{1}100\underbar{\textbf 0}x^*_- \xrightarrow{\Circled{4}} 
0\underbar{1}1001\underline{\textbf x^*_-} \label{eq:firstSteps_minus} \\
& 0\underbar{1}1001\underline{\textbf x^*_-} \xrightarrowdashed{T_{m-1} \text{ steps of } V_{\leq m - 1} \text{ in landscape of } \mathcal{C}^+_{n,\leq m - 1} }
0\underbar{\textbf 1}1001x^*_+ \label{eq:firstRec_minus} \\
& 0\underbar{\textbf 1}1001x^*_+ \xrightarrow{\Circled{5}}
00\underbar{\textbf 1}001x^*_+ \xrightarrow{\Circled{6}}
00000\underbar{\textbf 1}x^*_+ \xrightarrow{\Circled{7}}
000000\underline{\textbf x^*_+} \label{eq:secondSteps_minus} \\
& 000000\underline{\textbf x^*_+} \xrightarrowdashed{T_{m-1} \text{ steps of } V_{\leq m - 1} \text{ in landscape of } \mathcal{C}^-_{n,\leq m - 1} }
\underbrace{000000x^*_-}_{x^*(\mathcal{C}^-_{n,\leq m})} \label{eq:secondRec_minus} 
\end{align}
Finally, both the steepest ascent in the fitness landscape of $\mathcal{C}^+_{n,\leq m}$ starting from the assignment $x^*(\mathcal{C}^-_{n,\leq m})$ (\cref{eq:firstSteps,eq:firstRec,eq:secondSteps,eq:secondRec}) 
and the steepest ascent in the fitness landscape of $\mathcal{C}^-_{n,\leq m}$ starting from the assignment $x^*(\mathcal{C}^+_{n,\leq m})$ (\cref{eq:firstSteps_minus,eq:firstRec_minus,eq:secondSteps_minus,eq:secondRec_minus}) 
have lengths of $T_m = 7 + 2T_{m - 1}$ steps.
The recurrence of $T_1 = 7$ and $T_m = 7 + 2T_{m-1}$ is solved by $T_m = 7(2^m - 1)$.
\end{proof}

Combining \cref{thm:steepest-ascent} with that greedy local search follows steepest ascents, and using the facts that our instances are very sparse by \cref{prop:pw2}, and that they are oriented by \cref{prop:oriented}, we conclude that greed is slow on sparse graphs of oriented valued constraints:

\begin{theorem}\label{thm:greedy-runtime}
    Greedy local search can take $7(2^n-1)$ steps to find the unique local optimum in oriented binary Boolean \VCSP-instances $\mathcal{C}_{n,\leq n}^\pm$ on $6n$ variables with $6n$ unary constraints and $7n-1$ binary constraints, maximum degree $3$ and pathwidth $2$.
\end{theorem}

\section{Discussion}

There are two important structural features of our construction: that it is sparse (\cref{prop:pw2}) and that is is oriented (\cref{prop:oriented}).
Sparseness resolves in the negative the two open questions on the efficiency of greedy local search for \VCSP-instances of degree $3$ and treewidth $2$.
Given that all ascents are short for \VCSP{s} of degree $2$~\cite{KazThesis} and treewidth $1$~\cite{repCP}, the $\mathcal{C}^\pm_{n,\leq n}$ family of instances that are hard for greedy local search belong to the simplest class of instances where \emph{some} local search has the possibility to fail. 
That $\mathcal{C}^\pm_{n,\leq n}$ is an oriented \VCSP{} (\cref{prop:oriented}), reminds us that this failure of greedy local search is not due to the popular suspect of ``bad'' local peaks blocking the way to a hard to find global peak.
Oriented \VCSP{s} only produce semismooth fitness landscapes that are single peaked on each sublandscape~\cite{conditionallySmooth}:
there are no ``bad'' local peaks in this family of instances, there is the single global peak.
Further, in semismooth fitness landscapes, there is always a short ascent of Hamming distance from any initial assignments to unique peak~\cite{completelyUnimodal,evoPLS}.
A short ascents is always available in $\mathcal{C}^\pm_{n,\leq n}$, but greed stops us from find it.

Many other local search heuristics do not lose their way on oriented \VCSP{s}.
In contrast to \cref{thm:steepest-ascent}, consider how a local search method that chooses improving steps at random -- known as random ascent -- behaves on $\mathcal{C}_{n,\leq n}^-$ starting from any assignment.
Since there are $6n$ variables, any improving bit flip has a probability of at least $1/6n$ of being chosen as the next step. 
Thus, if $x_{(n,1)} \neq 0 (= x^*_{(n,1)})$ then after an expected number of at most $6n$ steps, it will be flipped to $0$. 
Given the oriented constraints, it cannot be flipped back for the rest of the run. 
Next, if -- at this point in the run -- the current assignment has $x_{(n,2)} \neq 0$ or $x_{(n,4)} \neq 0$ then in an expected number of about $2 \cdot 6n$ steps, $x_{(n,2)}$ and $x_{(n,4)}$ will also be flipped to $0$.
Given the oriented constraints and that the only in-edge to these variables is from $x_{(n,1)}$ which is now fixed to $0$, $x_{(n,2)}$ and $x_{(n,4)}$ will not be flipped back for the rest of the run.
This logic continues for $x_{(n,3)}$ and $x_{(n,5)}$, then $x_{(n,6)}$, then $x_{(n - 1,1)}$ and all other remaining variables following the arrows of the oriented constraints (i.e., in the topological order of the oriented \VCSP). 
This simple argument gives us a bound of $(6n)^2$ on the expected number of steps for random ascent to find the peak. 

The behavior of inadvertently fixing variables to their optimal state in the topological order of an oriented \VCSP{} is a feature of many local search methods, not just random ascent.
Kaznatcheev and Vazquez Alferez \textcite{conditionallySmooth} give a more careful analysis of random ascent on general oriented binary Boolean \VCSP{s}.
Applying their bound yields an expected number of at most $32n^2-18n$ steps to solve $\mathcal{C}_{n,\leq n}^\pm$.
For the task of solving general oriented \VCSP{s} they give quadratic bounds on the number of steps taken by, random ascent, simulated annealing, Zadeh's simplex rule, the Kerninghan-Lin heuristic, and various other local search methods. 
Thus, not only does greedy local search require an exponential number of steps on oriented \VCSP{s} of maximum degree $3$ and pathwidth $2$, but lots of other local search methods can solve these instances in quadratic time.
We conclude that, among local search methods, greed is slow.

%\printbibliography

%%
%% Bibliography
%%

%% Please use bibtex, 

\bibliography{SAextended}

\end{document}

%% file: img-and-andn.tex
\begin{figure}[!tb]
    \centering
    
    \begin{tikzpicture}[every node/.style={minimum size=30pt,font=\small}]
    \def\w{M_k} %the factor multiplying all weights (2^{7(k-1)+n})
    \def\s{S} %% multiplication so weights are integers
    \tikzmath{\xunit = 2.6; \yunit =3;}% the units for onw step along the axis
        \node[align=center,draw] (M) at (2*\xunit,-\yunit) {$M_k=6(2^k-2)$\\$S=2n+1$\\$s_k=n+1-k$};%box node
        
        %invisible nodes
        \node[draw, dashed, circle ] (inv) at (-2*\xunit,0) {\tiny $k+1,6$};
        
        %other module nodes
        %\node[draw, dashed, circle] (inv2) at (3*\xunit,\yunit) {\tiny $k-1,2$};%$k-1,2$
        %\node[draw,circle, dashed] (inv3) at (3*\xunit,-1*\yunit) {\tiny $k-1,4$};%$k-1,3$
        
        \node[draw, circle, label={[label distance=0cm, rotate=-30]-90:$-(2(\w+5)+1)\cdot\s$}] (ANDin) at (-\xunit,0) {$k,1$};
        \node[draw, circle, label={[label distance=0cm, rotate=20]90:$-(\w+4)\cdot\s-s_k$}] (AND1) at (0,\yunit) {$k,2$};%k,3
        \node[draw, circle, label={[label distance=0cm, rotate=-20]-90:$-(\w+5)\cdot\s$}] (AND2) at (0,-\yunit) {$k,4$};

        \node[draw, circle, label={[label distance=0cm,rotate=-20]90:$-(\w+3)\cdot \s$}] (XOR1) at (\xunit,\yunit) {$k,3$};
        \node[draw, circle, label={[label distance=0cm,rotate=20]-90:$-\s$}] (XOR2) at (\xunit,-\yunit) {$k,5$};
        \node[draw, circle, label={[label distance=0cm,rotate=20]-90:$-(\w+1)\cdot\s$}] (XORjoin) at (2*\xunit,0) {$k,6$};
        \node[draw, dashed, circle ] (XORout) at (3*\xunit,0) {\tiny $k-1,1$};
        
        % binary constrains
        \draw[->, thick] (ANDin) -- (AND1) node [midway, above, sloped] {$(\w+5)\cdot\s$};
        \draw [->, thick] (ANDin) -- (AND2)  node [midway, above, sloped] {$(\w+5)\cdot\s+s_k$};
    
        \draw[->, thick] (AND1) -- (XOR1) node [midway, above, sloped] {$(\w+4)\cdot\s$};
        \draw[->, thick] (AND2) -- (XOR2) node [midway, above, sloped] {$(\w+4)\cdot\s$};

        \draw[->, thick] (XOR1) -- (XORjoin) node [midway, above, sloped] {$(\w+2)\cdot\s$};
        \draw[->, thick] (XOR2) -- (XORjoin) node [midway, above, sloped] {$-(\w+2)\cdot\s$};

        \draw[->, dashed] (XORjoin) -- (XORout) node [midway, above, sloped] {$\w\cdot\s$};
    
        %dotted line
        \draw[->,dashed] (inv) -- (ANDin) node [midway, above=0.3cm] {$\overbrace{2(M_{k}+6)}^{M_{k + 1}} \cdot \s$};
        %\draw[dashed] (XORout) -- (inv2); %node [midway, above, sloped, fill=white] {$(34)M_{k-1}+2^{n-{(k-1)}}$}
        %\draw[dashed] (XORout) -- (inv3); %node [midway, above, sloped, fill=white] {$(34)M_{k-1}+2^{n}$}
        
    \end{tikzpicture}
    \caption{A gadget $\mathcal{C}^-_{n,k}$ with $M_k=6(2^k-2)$, $S = 2n + 1$, and $s_k = n + 1 - k$. 
    The constraints of the $k$th of $n$ gadgets are shown:
    the weights of unary constraints are next to their variables and the weights of binary constraints are above the edges that specify their scope.
    The orientation of the arcs is displayed, showing the instance is oriented.
    The dotted edges and vertices illustrate the connection to the neighboring gadgets. 
    %For the boundaries: the unary of $(n,1)$ is $(2\cdot M_n+11)\cdot S$ and there is no binary constraint $c_{(1,6),(0,1)}$.
    For the right boundary: there is no $0$th gadget and thus no constraint with scope $\{(1,6),(0,1)\}$.
    For the left boundary: both of $\mathcal{C}^\pm_{n,\leq m}$ have no constraint with scope $\{(m+1,6),(m,1)\}$ but $\mathcal{C}^+_{n,m}$ also changes the weight of the unary on $(m,1)$ to $c^+_{(m,1)} = c^-_{(m,1)} + c_{(m+1,6),(m,1))} = S$.
    }
    \label{fig:steepestGadget}
\end{figure}

%% file: img-landscape-zeros.tex
\scalebox{0.54}{

\begin{tikzpicture}[> = stealth, every path/.append style = {
        arrows = ->,
    }, steep/.style ={very thick}]
    \def \xDist {2cm};
    
    \node (b1) {$\underline{\mathbf{0}}00000$};
    
    \node (b2) [below = of b1] {$1\underline{\mathbf{0}}0\underline{0}00$};
    
    \node (b3) [below left= of b2] {$11\underline{\mathbf{0}}0\underline{0}0$};
    \node (b4) [below right= of b2] {$1\underline{0}01\underline{0}0$};
    
    \node (b5) [below left= of b3] {$111\underline{0}0\underline{\mathbf{0}}$};
    \node (b6) [below right= of b3] {$11\underline{0}1\underline{0}0$};
    \node (b7) [below right= of b4] {$100110$};
    
    \node (b8) [below left= of b5] {$111\underline{\mathbf{0}}01$};
    \node (b9) [below left = of b6] {$1111\underline{0}0$};
    \node (b10) [below left= of b7] {$11\underline{0}110$};
    
    \node (b11) [below right = of b8] {$1111\underline{\mathbf{0}}1$};
    
    \node (b12) [below right = of b11] {$11111\underline{\mathbf{1}}$};

    \node (b13) [below right= of b12] {$111110$};

    \path (b1) edge[steep] node [midway, right, fill=white] {$S$} (b2);
    \path (b1) edge[steep] node [midway, left, fill=white] {\Circled{1}} (b2);
    
    \path (b2) edge[steep] node [midway, right, fill=white] {$S-s_k$}  (b3);
    \path (b2) edge[steep] node [midway, left, fill=white] {\Circled{2}}  (b3);
    \path (b2) edge[red] node [midway, right, fill=white] {$s_k$} (b4);
    
    \path (b3) edge[steep] node [midway, right, fill=white] {$S$} (b5);
    \path (b3) edge[steep] node [midway, left, fill=white] {\Circled{3}} (b5);
    \path (b3) edge[red] node [midway, right, fill=white] {$s_k$} (b6);
    \path (b4) edge node [midway, right, fill=white] {$S-s_k$} (b6);
    \path (b4) edge node [midway, right, fill=white] {$(M_k+3)S$} (b7);
    
    \path (b5) edge[steep] node [midway, right, fill=white] {$\geq S$} (b8);
    \path (b5) edge[steep] node [midway, left, fill=white] {\Circled{4}} (b8);
    \path (b5) edge[red] node [midway, right, fill=white] {$s_k$} (b9);
    \path (b6) edge node [midway, right, fill=white] {$S$} (b9);
    \path (b6) edge node [midway, %right,
    fill=white] {$(M_k+3)S$}  (b10);
    \path (b7) edge node [midway, right, fill=white] {$S-s_k$} (b10);

    \path (b8) edge node [midway, left, fill=white] {\Circled{5}\;\;} (b11);
    \path (b8) edge[steep, red] node [midway, right, fill=white] {$s_k$} (b11);
    \path (b9) edge node [midway, %right,
    fill=white] {$(M_k+3)S$}  (b13);
    \path (b10) edge node [midway, right, fill=white] {$S$} (b13);

    \path (b11) edge[steep] node [midway, right, fill=white] {$S$} (b12);
    \path (b11) edge[steep] node [midway, left, fill=white] {\Circled{6}} (b12);
    
    \path (b12) edge[steep] node [midway, right, fill=white] {$\geq S$}(b13);
    \path (b12) edge[steep] node [midway, left, fill=white] {\Circled{7}}(b13);
\end{tikzpicture}

}

%% file: img-landscape-ones.tex
\scalebox{0.54}{

\begin{tikzpicture}[> = stealth, every path/.append style = {
        arrows = ->,
    }, steep/.style ={very thick}]
    \def \xDist {2cm};
    
    \node (b1) {$\underline{\mathbf{1}}11110$};
    
    \node (b2) [below = of b1] {$0\underline{1}1\underline{\mathbf{1}}10$};
    
    \node (b3) [below right= of b2] {$0\underline{1}10\underline{\mathbf{1}}0$};
    \node (b4) [below left= of b2] {$00\underline{11}10$};
    
    \node (b5) [below right= of b3] {$0\underline{1}100\underline{\mathbf{0}}$};
    \node (b6) [below left= of b3] {$00\underline{1}0\underline{1}0$};
    \node (b7) [below left= of b4] {$000\underline{1}10$};
    
    \node (b8) [below right= of b5] {$0\underline{\mathbf{1}}1001$};
    \node (b9) [below right = of b6] {$00\underline{1}000$};
    \node (b10) [below right= of b7] {$0000\underline{1}0$};
    
    \node (b11) [below left = of b8] {$00\underline{\mathbf{1}}001$};
    
    \node (b12) [below left = of b11] {$00000\underline{\mathbf{1}}$};

    \node (b13) [below left= of b12] {$000000$};

    \path (b1) edge[steep] node [midway, right, fill=white] {$S-s_k$} (b2);
    \path (b1) edge[steep] node [midway, left, fill=white] {\Circled{1}} (b2);
    
    \path (b2) edge[steep] node [midway, right, fill=white] {$S$}  (b3);
    \path (b2) edge[steep] node [midway, left, fill=white] {\Circled{2}}  (b3);
    \path (b2) edge[red] node [midway, right, fill=white] {$s_k$} (b4);
    
    \path (b3) edge[steep] node [midway, right, fill=white] {$S$} (b5);
    \path (b3) edge[steep] node [midway, left, fill=white] {\Circled{3}} (b5);
    \path (b3) edge[red] node [midway, right, fill=white] {$s_k$} (b6);
    \path (b4) edge node [midway, right, fill=white] {$S$} (b6);
    \path (b4) edge node [midway, %left,
    fill=white] {$(M_k+3)S$} (b7);
    
    \path (b5) edge[steep] node [midway, right, fill=white] {$\geq S$} (b8);
    \path (b5) edge[steep] node [midway, left, fill=white] {\Circled{4}} (b8);
    \path (b5) edge[red] node [midway, right, fill=white] {$s_k$} (b9);
    \path (b6) edge node [midway, right, fill=white] {$S$} (b9);
    \path (b6) edge node [midway, %right,
    fill=white] {$(M_k+3)S$}  (b10);
    \path (b7) edge node [midway, left, fill=white] {$S$}  (b10);

    \path (b8) edge node [midway, left, fill=white] {\Circled{5}\;\;} (b11);
    \path (b8) edge[steep, red] node [midway, right, fill=white] {$s_k$} (b11);
    \path (b9) edge node [midway, right, fill=white] {$(M_k+3)S$}  (b13);
    \path (b10) edge node [midway, left, fill=white] {$S$} (b13);

    \path (b11) edge[steep] node [midway, right, fill=white] {$S$} (b12);
    \path (b11) edge[steep] node [midway, left, fill=white] {\Circled{6}} (b12);
    
    \path (b12) edge[steep] node [midway, right, fill=white] {$\geq S$}(b13);
    \path (b12) edge[steep] node [midway, left, fill=white] {\Circled{7}}(b13);
\end{tikzpicture}

}